\documentclass[conference]{IEEEtran}

\usepackage{eulervm}

\usepackage[english]{babel}
\usepackage{amsmath,amssymb}
\usepackage{times}
\usepackage{relsize}
\usepackage{graphicx}
\usepackage{url}
\usepackage{booktabs}
\usepackage{multirow}
\usepackage{mparhack}
\usepackage{subfigure}
\usepackage{authblk}
\usepackage{amsthm}
\usepackage{mathrsfs}

\usepackage[noend]{algorithmic}
\usepackage[linesnumbered,ruled,vlined]{algorithm2e}

\usepackage{array}
\usepackage{cite}
\usepackage{eqparbox}
\usepackage{mdwmath}

\usepackage{epsfig}
\usepackage{xcolor}

\usepackage{mathtools}
\usepackage{bm}
\usepackage{bbold}
\usepackage[all]{xy}
\usepackage{etoolbox}
\DeclareMathAlphabet{\mathantt}{OT1}{antt}{li}{it}
\DeclareMathAlphabet{\mathpzc}{OT1}{pzc}{m}{it}

\DeclarePairedDelimiter\norm{\lVert}{\rVert}%
\usepackage{accents}

\newtheorem{theorem}{Theorem}

\theoremstyle{remark}
\newtheorem{lemma}[theorem]{Lemma}

\newtheorem{definition}[theorem]{Definition}
\newtheorem{notation}[theorem]{Notation}

\DeclareFontFamily{OT1}{pzc}{}
\DeclareFontShape{OT1}{pzc}{m}{it}%
  {<-> s * [1.1] pzcmi7t}{}
\DeclareMathAlphabet{\mathpzc}{OT1}{pzc}%
                     {m}{it}

\def\R{\mathcal{R}}
\def\I{\mathcal{I}}
\def\J{\mathcal{J}}

\def\C{\mathcal{C}}

\DeclareMathOperator{\argmax}{\arg\max}

\newcommand*\diff{\mathop{}\!\mathrm{d}}

\usepackage{geometry}
\title{Joint CoMP-Cell Selection and Resource Allocation in
Fronthaul-Constrained C-RAN}

\author[1]{Lei You}
\author[1,2]{Di Yuan}
\affil[1]{{\small Department of Information Technology, Uppsala University, Sweden}}
\affil[2]{{\small Department of Science and Technology, Link\"{o}ping University, Sweden}}
\affil[ ]{\small\texttt{\{lei.you; di.yuan\}@it.uu.se}\/}

\begin{document}

\maketitle

\begin{abstract}
Cloud-based Radio Access Network (C-RAN) is a promising architecture for future
cellular networks, in which Baseband Units (BBUs) are placed at a centralized
location, with capacity-constrained fronthaul connected to multiple distributed
Remote Radio Heads (RRHs) that are far away from the BBUs. The centralization of
signal processing enables the flexibility for coordinated multi-point
transmission (CoMP) to meet high traffic demand of users. We
investigate how to jointly optimize CoMP-cell selection and base station resource
allocation so as to enhance the quality of service (QoS), subject to the
fronthaul capacity constraint in orthogonal frequency-division multiple access
(OFDMA) based C-RAN\@. The problem is proved to be $\mathcal{NP}$-hard in this
paper. To deal with the computational complexity, we derive a partial optimality
condition as the foundation for designing a cell-selection algorithm. Besides, we
provide a solution method of the optimum of the time-frequency
resource allocation problem without loss of fairness on the QoS enhancement of
all users. The simulations show good performance of the proposed algorithms for
jointly optimizing the cell selection and resource allocation in a C-RAN, with
respect to QoS.
\end{abstract}

\section{Introduction}

\subsection{Background}

Mobile networks are evolving rapidly in terms of coverage, capacity and new
features, continuously pushed by new requirements related to latency, traffic
volumes and data rates~\cite{EricssonAB:2015tv}. The operators are seeking for
promising ways to increase the
flexibility of cellular infrastructures so as to simplify the deployment and
management of the network integrated with different communication technologies,
e.g.~mmWave, Li-Fi, coordinated multipoint transmission (CoMP),
etc.~\cite{EricssonAB:2015tv}. Cloud-based Radio Access Network (C-RAN),
referring to the virtualization of base station functionalities by means of
cloud computing, results in a novel cellular architecture that cost-efficiently
enables centralization and cloudization of large-scale cooperative signaling
processing in a network-wide manner~\cite{Peng:2015bx}, and thus reducing the
overall network complexity~\cite{Peng:2015gc} in respect of management and
deployment. 

In a C-RAN, the functions of Baseband Units (BBUs) and Remote Radio Heads (RRHs)
are redefined, with some of the BBU processing functions being shifted to the
RRH, which leads to a change in the BBU and RRH
architecture~\cite{Anonymous:Mvikarf0}.  The BBUs are aggregated in a pool and
co-located in a Radio Cloud Center (RCC), and the RRHs that are separately
distributed away from the RCC form Radio Remote Systems (RRSs). The link
connecting a BBU with an RRH is referred to as fronthaul.  One advantage of C-RAN lies in
its ability to implement the CoMP transmission among multiple RRHs, based on its
centralization of signal
processing~\cite{Anonymous:Mvikarf0,ViaviSolutionsInc:2015th,Artuso:2015hm}. On
the other hand, CoMP may result in more data traffic on fronthaul, which means
that the CoMP performance depends on the capacity of
fronthaul~\cite{Anonymous:Mvikarf0,ViaviSolutionsInc:2015th,Artuso:2015hm}.
With the limited fronthaul capacity and bit-rate demands of user
equipments (UEs), the CoMP performance benefits from optimizing the association
pattern between RRHs and UEs.  

\subsection{Motivation}

Several literature has focused on the CoMP techniques in C-RAN\@.
In~\cite{Zhang:2015kr}, the authors studied the CoMP-based interference
mitigation in heterogeneous C-RANs deployed with small cells.
In~\cite{Davydov:2013cf}, the authors investigated the Joint Transmission (JT)
CoMP performance in C-RANs with large CoMP cluster sizes. Also, it is shown
by~\cite{Davydov:2013cf} that CoMP transmission can be efficiently and
effectively implemented based on the cooperation of a limited set of stations
forming a so-called \textit{``CoMP cluster''} in a C-RAN\@. The authors
in~\cite{Beylerian:2016do} investigated the resource allocation of CoMP
transmission in C-RAN, and proposed a fairness-aware user-centric scheme for
enhancing the network coverage and achievable rate. In~\cite{MugenPeng:2014ec},
state-of-the-art and challenges of heterogeneous C-RANs are surveyed. The
authors showed that limited fronthaul capacity affects the CoMP performance,
pointing out that optimal resource allocation solutions call for investigation 
under the fronthaul capacity constraint.  In~\cite{Ortin:2016jq}, the
authors studied jointly cell-selection and resource allocation problems, in
C-RANs of non-CoMP case. In~\cite{Abdelnasser:2016ho}, the resource allocation
problem was studied in OFDMA-based C-RANs, with the framework of small cell
underlaying a macro cell\@.

We remark that our work is motivated by the research mentioned above, as all of
them stressed the benefits on network performance improvement by optimization of
cell selection as well as resource allocation in C-RANs. On the other hand, 
so far we have not found literature that studies how
to jointly optimize the cell selection and network-wide time-frequency resource
allocation of CoMP transmission in C-RANs\@. From our perspective, in the CoMP
scenarios, it is more crucial to optimize the time-frequency resource allocation
subject to the constraints on QoS demands compared to the non-CoMP case, as CoMP
introduces more dependencies of resource consumption and control signaling 
among cells/stations involved in cooperations. Also, under the capacity limit of
fronthaul, the QoS satisfaction is highly affected by the CoMP-cell selection
along with the resource allocation strategies. 

\subsection{Contributions}
In this paper, we study how to jointly optimize the CoMP-cell selection and the
time-frequency resource allocation among cells, subject to the limit of
fronthaul capacity, in order to maximize the fairness-aware QoS. We prove the
$\mathcal{NP}$-hardness of the problem, and provide theoretical insights as
foundation of designing an efficient cell selection algorithm. As for the
resource allocation, we derive a solution method with respect to optimizing QoS
performance, subject to the QoS fairness constraints among all users. We proved
that the solution method for resource allocation achieves the optimum, under any
chosen CoMP-cell selection. Also, we show theoretically that the proposed
CoMP-cell selection algorithm and the resource allocation method are naturally
combined with each other, which is shown numerically to achieve good performance
on fairness-aware QoS.

\section{System Model}
\label{sec:sys_model}

\subsection{General Description}
The setting of the system model is coherent with~\cite{Anonymous:Mvikarf0}.
We consider a C-RAN with a centralized RCC, which has $n$ BBUs in the BBU pool.
 Each BBU is connected with several RRHs by fronthaul. This BBU and its
connected RRHs form a so-called \textit{``C-RAN Cluster''}, and is referred to
as ``cluster'' in this paper for simplicity.  The system model is illustrated in
\figurename~\ref{fig:sys_model}. We consider JT-CoMP\@. Within a cluster, the
BBU calculates the coordinated beamformer for each RRH and all co-clustered RRHs
can be jointly coordinated to serve UEs. Without loss of generality, we assume
that each RRH is located on one Base Station (BS)\@. Downlink transmission is
studied in this paper.

\begin{figure}[!h]
    \centering
    \includegraphics[width=\linewidth]{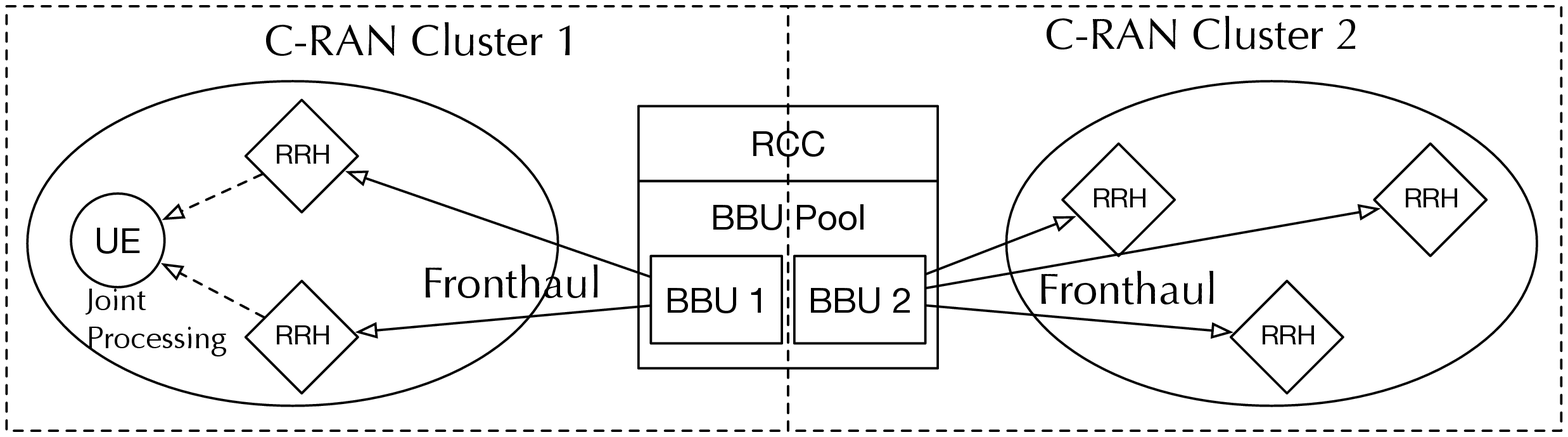}
    \caption{System model illustration.}
\label{fig:sys_model}
\end{figure}

\subsection{Basic Notations}
Denote by $\C=\{1,2,\ldots,n\}$ the set of C-RAN clusters. Denote by
$\R=\{1,2,\ldots,m\}$ the set of RRHs in the C-RAN\@. Denote by
$\J=\{1,2,\ldots,q\}$ the set of UEs. Since each cluster has only one BBU and
there is no shared BBU among clusters, we also use $1,2,\ldots,n$ to refer to
the corresponding RRHs for the $n$ clusters, respectively. For the same reason,
we use $1, 2,\ldots,m$ to refer to the corresponding BSs located with the $n$
RRHs as well as the corresponding fronthaul, respectively. 
To avoid being ambiguous, the network entities that we refer to by using the
indexes, will be explicitly clarified, in accordance with the context.
Denote by $c_i$ the capacity of the fronthaul connected to RRH $i$
($i\in\R$). Denote by $\ell(i)$ and $\ell(j)$ the cluster in which RRH $i$
($i\in\R$) and UE $j$ ($j\in\J$) located, respectively.
For each UE $j$, denote by $\I_j$ the set of BSs/RRHs serving UE $j$ via CoMP\@.
We use the symbol $\circ$ to denote function composition.

\subsection{CoMP Transmission}

Denote by $p_i$ the transmit power of BS $i$, $i\in\R$. Denote by $\bm{h}_{ij}$
the channel gain between BS $i$ and UE $j$. Denote by $\bm{w}_{i}$ the precoder
of BS $i$. Let $\bm{x}$ be the channel input symbol sent by the cooperating BSs
$\I_j$. Entity $\bm{x}_k$ denotes the channel input symbol sent by the other
cells that are not cooperatively serving UE $j$. The received channel output at
UE $j$ can be written as 
\begin{equation}
    \bm{s} =
    \sum_{i\in\I_j}\sqrt{p}_{i}\bm{h}^{\mathsf{H}}_{ij}\bm{w}_{i}\bm{x}+\sum_{k\in\R\backslash\I_j}\sqrt{p}_k\bm{h}^{\mathsf{H}}_{kj}\bm{w}_{k}\bm{x}_{k}+\sigma
\end{equation}

Assuming that $\bm{x}$ and $\bm{x}_k$ $k\in\R\backslash\I_j$ are independent
zero-mean random variables of unit variance, the SINR of UE $j$ is given by
the equation as below~\cite{Nigam:2014cd},
\begin{equation}
    \gamma_j = \frac{\left|\sum_{i\in\I_j}\sqrt{p}_{i}\bm{h}^{\mathsf{H}}_{ij}
    \bm{w}_i\right|^2}{\sum_{k\in\R\backslash\I_j}p_{k}I_{kj}+\sigma^2}
\label{eq:sinr}
\end{equation}
where we have $I_{kj} = |\bm{h}^{\mathsf{H}}_{kj}\bm{w}_{k}|^2 \rho_k$ and
$\rho_k$ is the indicator of that whether $j$ receives interference from BS
$k$.

For the sake of presentation, we use the binary matrix
$\bm{\kappa}\in{\{0,1\}}^{m\times q}$ as the indicator for the association
relationship between BSs and UEs, by defining $\I_j$ ($j\in\J$) as
a mapping of $\bm{\kappa}$ to a set of BS(s), as shown
in~\eqref{eq:kappa}. The obligation of each cluster for providing service to UEs
is clarified in~\eqref{eq:obligation}.
% \begin{equation}
%     \kappa_{ij} = 1 \Leftrightarrow i\in\I_j\cap\R_{\ell(j)};~
%     \kappa_{ij}\in\{0,1\}~i\in\R~j\in\J
% \label{eq:kappa}
% \end{equation}

\begin{equation}
    \I_j:{\{0,1\}}^{m\times q}\rightarrow \mathbb{2}^{\{1,2,\ldots,m\}}:\bm{\kappa}\mapsto\{i:\kappa_{ij}=1\}
\label{eq:kappa}
\end{equation}
\begin{equation}
    \kappa_{ij}=0 \quad \ell(i)\neq\ell(j)
\label{eq:obligation}
\end{equation}

\subsection{Resource Allocation}
Consider in OFDMA any fixed short period $T=[\tau_s,\tau_e]$ of a few hundred
milliseconds, with $\tau_s$ and $\tau_e$ be the starting and ending time points,
respectively. Without loss of generality, we use the term ``Resource Unit''
(RU) as the minimum unit for OFDMA time-frequency resource allocation. Suppose
that in each BS the total number of RUs is $M$.  For the cell(s) serving UE $j$
($j\in\J$), we denote $\alpha_j$ as the proportion of allocated RUs for
transmission to UE $j$ over the total number $M$ of RUs in each BS\@. (Note that
in JT-CoMP, the BSs that are cooperatively serving one UE use the same
time-frequency resource for transmission to the UE, and thus there is no need to
use index $\alpha_{ij}$ and $\alpha_{kj}$ for differentiating BS $i$ and $k$,
$i,k\in\I_j$.) For the network-wide resource allocation, we have $\bm{\alpha} =
{[\alpha_1,\alpha_2,\ldots,\alpha_q]}^{\mathsf{T}}$. There is no 
specific bias on resource allocation for transmission, and thus each RU has the
same opportunity used by BS for serving UE\@. The term
$\sum_{j\in\I_k}\alpha_j$ computes the proportion of occupied RUs in BS $k$. For any
RU that is being used for transmission, the probability that it is interfered by
another BS $k$ (assuming $k$ is not a cooperative BS in this transmission) is
$\sum_{j\in\I_k}\alpha_j$. For characterizing the influence of resource
allocation on the inter-cell interference, we let $\rho_k$ be the mapping
$\rho_k:\mathbb{R}_{+}^{q}\rightarrow\mathbb{R}_+:
\bm{\alpha}\mapsto\sum_{j\in\J_k}\alpha_j$ in~\eqref{eq:sinr}.  Also, $\rho_k$
is named as the load of BS $k$, for any $k\in\R$.

Denote by $B$ the bandwidth per RU\@. The achievable bit rate of a (JT) link to
UE $j$ can be calculated by $C_j:\mathbb{R}_{+}^{q}\rightarrow\mathbb{R}:\bm{\alpha}\mapsto
MB\log(1+\gamma_j\circ\bm{\rho}(\bm{\alpha}))$ according to the Shannon's
capacity. We assume the data traffic requested by UE $j$ $(j\in\J)$ takes place
with density $t_j$ over time period $T$, and
$\bm{t}=[t_1,t_2,\ldots,t_q]$. As for
ensuring the allocated RUs to UE $j$ is sufficient for transmission, we should
have in~\eqref{eq:coupling} that
\begin{equation}
    \alpha_{j}\geq
    \int_{T}\frac{\eta_{j}t_j(\tau)}{C_j(\bm{\alpha},\bm{\kappa})}\diff\tau
\label{eq:coupling}
\end{equation}
where $\eta_j$ is UE $j$'s bit rate scaling parameter over the time period $T$.
In other words, $\int_{T}\eta_j t_{j}(\tau)\diff\tau$ is the data traffic volume
transmitted to UE $j$ scaled by $\eta_j$ over period $T$. Network-widely, we have
$\bm{\eta}= {[\eta_1,\eta_2,\ldots, \eta_n]}^{\mathsf{T}}$.  As for the
constraint of limited resource and fronthaul capacity,~\eqref{eq:limited} holds,
where $\bar{\rho}$ is the maximum resource limit of BSs.
\begin{equation}
    \sum_{j\in\J_i}\int_{T}\eta_{j}t_j(\tau)\diff\tau\leq c_i\textnormal{ and
    }\rho_i(\bm{\alpha},\bm{\kappa}) \leq
    \bar{\rho}~~~i\in\R
\label{eq:limited}
\end{equation}

Note that $\eta_j$ reflects the ``\textit{QoS satisfaction}'' of UE $j$. In
other words, $\eta_j$ indicates how well the traffic demand of UE $j$ is
satisfied, under the constraints~\eqref{eq:coupling} and~\eqref{eq:limited}.
Besides, we take into consideration the fairness of QoS satisfaction among
UEs. The metric \textit{Jain's Fairness} is used for evaluation, i.e.
$J:\mathbb{R}_{+}^q\rightarrow\mathbb{R}:\bm{\eta}\mapsto\norm{\bm{\eta}}_1^2/q\norm{\bm{\eta}}^2$.
It is guaranteed by~\eqref{eq:fairness} that the best fairness among UEs is
achieved.
\begin{equation}
    J(\bm{\eta})=1
\label{eq:fairness}
\end{equation}

\subsection{Problem Formulation}
The investigated problem is to maximize the QoS by CoMP-cell
selection and sufficient RU allocation ensured by~\eqref{eq:coupling}, subject
to the capacity/resource limits in~\eqref{eq:limited}, and the fairness constraint
in~\eqref{eq:fairness}. The problem is formulated in
\eqref{eq:p0}.
\begin{equation}
\max_{\bm{\kappa},\bm{\alpha},\bm{\eta}} \norm{\bm{\eta}}_1\quad\textnormal{subject to }
\eqref{eq:kappa}-\eqref{eq:fairness}
\label{eq:p0}
\end{equation}

\section{Computational Complexity}
\label{sec:computation}

% In this section, we prove the computational complexity of the problem
% in~\eqref{eq:p0}.
\begin{theorem}
The problem in~\eqref{eq:p0} is $\mathcal{NP}$-hard.
\label{thm:NP-hard}
\end{theorem}
\begin{proof}
We prove the theorem by a polynomial-time reduction from the 3-satisfiability
(3-SAT) problem that is $\mathcal{NP}$-complete. Consider a 3-SAT problem with
$N_1$ Boolean variables $b_1$, $b_2$, $\ldots$, $b_{N_1}$, and $N_2$ clauses. A
Boolean variable or its negation is referred to as a literal, e.g.  $\hat{b}_i$
is the negation of $b_i$. A clause is composed by a disjunction of exactly three
distinct literals, e.g.  $(b_1\vee b_2 \vee\hat{b}_3)$ is an example of clause.
The 3-SAT problem amounts to determining whether or not there exists an
assignment of $\mathsf{true}$/$\mathsf{false}$ to the variables, such that all
clauses would be satisfied. To make the reduction from 3-SAT to our problem
in~\eqref{eq:p0}, we construct a specific network scenario as follows. Suppose
we have $N_1+N_2+1$ UEs in total, denoted by
$u_{0},u_{1},u_{2},\ldots,u_{N_1+N_2}$, respectively. Also, we have in total
$2N_1+N_2$ BSs, denoted by $a_1,a'_1,a_2,a'_2,\ldots a_{N_1},a'_{N_1}$, and
$A_{0}, A_{N_1+1},A_{N_1+2},\ldots,A_{N_1+N_2}$, respectively. For each $u_{i}$
($1\leq i\leq N_1$), we set $\R_{\ell(u_{i})}=\{a_i,a'_i\}$. For each $u_{j}$
($N_1<j\leq N_1+N_2)$, we set $\R_{\ell(u_{j})}=\{A_j\}$. Besides, $\R_{\ell(u_{0})}$ is set
to $\{A_0\}$. Let $p_{A_0}=N_1+1$. For $1\leq i\leq N_1$ and $N_1<j\leq
N_1+N_2$, let $p_{a_i}=p_{A_j}=1$.  For simplicity, we use the term
``\textit{gain value}'' to refer to $|\bm{h}^{\mathsf{H}}_{ij}\bm{w}_{i}|$ shown
in~\eqref{eq:sinr}.  For any UE $u_i$ $(1\leq i\leq N_1+N_2)$, the gain values of
$a_i$ and $a'_i$ equal to $1.0$. For UE $u_0$, the gain values of all $a_i$ and
$a'_i$ ($1\leq i\leq N_1$) equal to $1.0$. Besides, the vector
$\bm{h}^{H}_{a_i,u_i}\bm{w}_{a_i}$ is orthogonal to
$\bm{h}^{H}_{\hat{a}_i,u_i}\bm{w}_{\hat{a}_i}$ for $1\leq i\leq N_1$. For any $j$
($1\leq j\leq N_1$), $u_{N_1+j}$ has the gain value $1.0$ from the BSs that
represent the literals in clause $j$. The gain value from
$A_0$ to $u_0$ is $1.0$. In addition, from $A_{i}$ to $u_{i}$ ($N_1<i\leq
N_1+N_2)$), the gain value is $3.0$. Gain values between any other BS-UE
pair are negligible. The noise power $\sigma^2$ is $1.0$. In addition, the total
traffic demand within the time period $T$ for each UE is $1.0$. We normalize the data traffic
of UEs within the time period $T$ by $B\times M$, and the normalized demands are
uniformly set to $1.0$. The fronthaul capacity is set to
be sufficient with respect to this user demand. 

First, we note that each UE $j$ ($0\leq j\leq N_1+N_2+1$) should be served by at
least one BS, otherwise $C_j$ equals to $0$ and the
constraint~\eqref{eq:coupling} would be violated. Thus, $A_0$ is serving $u_0$
and $A_{N_1+1}, A_{N_1+2}, \ldots, A_{N_1+N_2}$ are serving
$u_{N_1+1},u_{N_1+2},\ldots,u_{N_1+N_2}$, respectively.  Second, it can be
verified that $u_{i}$ ($1\leq i\leq N_1$) can only be served by exactly
one BS in $\R_{\ell(u_{i})}=\{a_{i},\hat{a}_{i}\}$. This is because, if
$u_{i}$ is served by both $a_{i}$ and $\hat{a}_{i}$, then 
the BS $A_0$ would be overloaded ($\rho_{A_0}>1$) due to the interference
received from all other BSs, and thus the maximum resource limit constraint in
~\eqref{eq:limited} would be violated. Besides, for each clause, the three
corresponding cells (e.g.~for a clause
$(b_1\vee b_2\vee \hat{b}_3)$ the corresponding cells are $a_1$, $a_2$ and
$\hat{a}_3$) cannot be all active in serving UEs. Otherwise, the cell that is
serving the UE corresponding to this clause would be overloaded, thus violating
the maximum resource limit constraint in~\eqref{eq:limited}. 

Now suppose there
is an association that is feasible. For each Boolean variable $b_i$, we set
$b_i=\mathsf{true}$ if $a'_{i}$ is serving UE $u_{i}$\@. Otherwise, $u_{i}$ must
be served by $a_{i}$ and we set $b_i=\mathsf{false}$. For each clause, it is
satisfied only if at least one of its literals is with the value $\mathsf{true}$.
As discussed above, a feasible solution of the constructed problem cannot
have all the corresponding three BSs been in the status of serving UEs, which
means that we have at least one of the three BSs been in the idle status.
Therefore, a feasible solution of the constructed problem is corresponding to
the 3-SAT problem instance. Hence the conclusion. 
\end{proof}

% Theorem~\ref{thm:NP-hard} implies that no low-complexity and exact algorithm can
% be expected, regardless of how clever algorithms can be designed, unless
% $\mathcal{P}=\mathcal{NP}$.

\section{Solution Approach}
\label{sec:solution}

We show theoretical insights of the joint optimization problem on cell selection
and resource allocation in this section. We provide a scheme to achieve global
optimal BS resource allocation, with any fixed BS-UE association pattern.
Furthermore, we derive partial optimality condition for CoMP-cell selection. An 
algorithm for solving the problem in~\eqref{eq:p0} is then proposed, based
on the theoretical foundations mentioned above.

\subsection{Optimal Resource Allocation}

For clarity, we define the following notations, which are used throughout
the remaining of this paper.

\begin{notation}
    For any mapping
    $\mathpzc{M}:\mathbb{R}^q_{+}\rightarrow\mathbb{R}^{q}_{+}:\bm{\alpha}_{\geq
    0}\mapsto\mathpzc{M}(\bm{\alpha}_{\geq 0})$,
    denote $\mathpzc{M}^{(k+1)}=\mathpzc{M}^{(k)}\circ\mathpzc{M}$ for any
    $k\geq 1$. 
\end{notation}

\begin{notation}
    Denote the mapping
    $\mathpzc{T}_{\bm{\kappa},j}:\mathbb{R}_{+}^{q+1}\rightarrow\mathbb{R}_{+}:
    [\bm{\alpha},\eta_j]\mapsto\int_{T}\eta_{j}t_j(\tau)\big/C_j(\bm{\alpha},\bm{\kappa})\diff\tau$.
\end{notation}

\begin{notation}
    Denote the mapping
    $\mathpzc{H}_{\bm{\kappa}}:\mathbb{R}_{+}^{q+1}:
    \rightarrow\mathbb{R}^{q}_{+}:[\bm{\alpha},\eta_j]\rightarrow{[\mathpzc{T}_{\bm{\kappa},1}(\bm{\alpha},\eta_j),\mathpzc{T}_{\bm{\kappa},2}(\bm{\alpha},\eta_j),\ldots,\mathpzc{T}_{\bm{\kappa},q}(\bm{\alpha},\eta_j)]}^{\mathsf{T}}$.
\end{notation}

\begin{notation}
    Denote by $\norm{\bm{\alpha}_{\geq 0}}_{\bm{\kappa}}$ a norm of $\bm{\kappa}$ with
    respect to $\bm{\kappa}$, such that
    $\norm{\bm{\alpha}_{\geq
    0}}_{\bm{\kappa}}=\norm{\bm{\kappa}\bm{\alpha}_{\geq 0}}_{\infty}$,
    $\bm{\alpha}_{\geq 0}\in\mathbb{R}^{q}_{+}$.
\end{notation}

Theorem~\ref{thm:opt_alpha} provides solution for achieving the optimal resource
allocation under any BS-UE association $\bm{\kappa}'$. Before giving the proof,
we derive several lemmas, providing theoretical insights
of~\eqref{eq:opt_alpha},~\eqref{eq:alpha_kappa}, and~\eqref{eq:alpha_kappa2}.
Theorem~\ref{thm:opt_alpha} is then proved based on these lemmas.

\begin{theorem}
    [\small{Optimal Resource Allocation}] For any
    $\bm{\kappa}'\in{\{0,1\}}^{m\times q}$, consider the
mapping $\bm{\alpha}^{*}:{\{0,1\}}^{m\times
q}\rightarrow\mathbb{R}_{++}^{q}:\bm{\kappa}'\mapsto\argmax_{\bm{\alpha}}\{\norm{\bm{\eta}}_1:
\eqref{eq:coupling}-\eqref{eq:fairness}, \bm{\kappa}=\bm{\kappa}'\}$. Then 
\begin{equation}
    \alpha^{*}_j(\bm{\kappa}')=\min\left\{\alpha_j(\bm{\kappa}'),
    \hat{\alpha}_j(\bm{\kappa}')\right\} \quad j\in\J
\label{eq:opt_alpha}
\end{equation}
where 
\begin{equation}
    \alpha_j(\bm{\kappa}')=\lim_{k\rightarrow\infty}\bar{\rho}
    {{\mathpzc{T}^{(k)}_{\bm{\kappa}',j}
    (\bm{\alpha}_{\geq 0},1)}\bigg/\norm{\mathpzc{H}^{(k)}_{\bm{\kappa}'}
    (\bm{\alpha}_{\geq 0},1)}_{\bm{\kappa}'}}
\label{eq:alpha_kappa}
\end{equation}
and
\begin{equation}
    \hat{\alpha}_j(\bm{\kappa}')=\lim_{k\rightarrow\infty}{\mathpzc{T}^{(k)}_{\bm{\kappa}',j}
    (\bm{\alpha}_{\geq 0},\nu)}
\label{eq:alpha_kappa2}
\end{equation}
with 
$\nu=\min_{i\in\R}\left\{c_i\big/\int_{T}\sum_{j\in\J_i}t_j(\tau)\diff\tau\right\}$
and $\bm{\alpha}_{\geq 0}\in\mathbb{R}_{+}^{q}$.
\label{thm:opt_alpha}
\end{theorem}

\begin{lemma}
The limits in~\eqref{eq:alpha_kappa} and~\eqref{eq:alpha_kappa2} exist.
\label{lma:SIF}
\end{lemma}
\begin{proof}
The proof for~\eqref{eq:alpha_kappa2} is based on the conclusions
that~\eqref{eq:alpha_kappa2} falls into the category of standard interference
function (SIF)~\cite{Yates:1995eh}, with respect to the variable
$\bm{\alpha}_{\geq 0}$, as proved in~\cite{You:2016uy}. An SIF converges to a
fixed point that is unique.  The proof for~\eqref{eq:alpha_kappa} is based on
Theorem~1 in~\cite{Krause:2001wd}, that the normalized mapping
${{\mathpzc{H}_{\bm{\kappa}'}(\bm{\alpha}_{\geq
0},1)}\big/\norm{\mathpzc{H}_{\bm{\kappa}'} (\bm{\alpha}_{\geq
0},1)}_{\bm{\kappa}'}}$ converges to a unique fixed point (and thus
${{\mathpzc{T}_{\bm{\kappa}',j}(\bm{\alpha}_{\geq
0},1)}\big/\norm{\mathpzc{H}_{\bm{\kappa}'} (\bm{\alpha}_{\geq
0},1)}_{\bm{\kappa}'}}$ for $j\in\J$). Following the conclusion
above, the fixed points of $\alpha_j(\bm{\kappa}')$ and
$\hat{\alpha}_{j}(\bm{\kappa}')$ are computed by~\eqref{eq:alpha_kappa}
and~\eqref{eq:alpha_kappa2}, respectively.
\end{proof}

\begin{lemma}
$\norm{\bm{\rho}(\bm{\alpha}(\bm{\kappa}'),\bm{\kappa}')}_{\infty}=\bar{\rho}$.
\label{lma:bar_rho}
\end{lemma}
\begin{proof}
Suppose $\mathsf{fixp}\in\mathbb{R}_{+}^{q}$ is the fixed point of
${{\mathpzc{T}_{\bm{\kappa}',j}
(\mathsf{var},1)}\big/\norm{\mathpzc{H}_{\bm{\kappa}'} (\mathsf{var},
1)}_{\bm{\kappa}'}}$, with respect to $\mathsf{var}$. By Theorem~1
in~\cite{Krause:2001wd}, there exists $\lambda>0$, such that
${\mathpzc{T}_{\bm{\kappa}',j} (\mathsf{fixp},1)}=\lambda\mathsf{fixp}$, with
$\norm{\mathsf{fixp}}_{\bm{\kappa}'}=1$. 
Thus, for the function ${{\mathpzc{T}_{\bm{\kappa}',j}
(\mathsf{var},1)}\big/\norm{\mathpzc{H}_{\bm{\kappa}'} (\mathsf{var},
1)\frac{1}{\bar{\rho}}}_{\bm{\kappa}'}}$, we have
$\norm{\frac{1}{\bar{\rho}}\bm{\alpha}(\bm{\kappa}')}_{\bm{\kappa}'}=1$ at its fixed point
$\bm{\alpha}(\bm{\kappa}')$, which leads to
$\norm{\bm{\kappa}'\bm{\alpha}(\bm{\kappa}')}=\bar{\rho}$. By the definition of
the mapping $\bar{\rho}$ in Section~\ref{sec:sys_model}, we have
$\norm{\bm{\rho}(\bm{\alpha}(\bm{\kappa}'),\bm{\kappa}')}_{\infty}=\bar{\rho}$.
Hence the conclusion.
% According to Lemma~\ref{lma:SIF},  
% $\norm{\bm{\rho}(\bm{\alpha}(\bm{\kappa}'),\bm{\kappa}')}_{\infty}\alpha_j(\bm{\kappa}')=\mathpzc{T}_{\bm{\kappa}',j}
% (\bar{\rho}\bm{\alpha}(\bm{\kappa}'),1)$
% holds for any $j\in\J$. As the function $\mathpzc{T}_{\bm{\kappa}',j}
% (\mathsf{var},1)$ is an SIF in $\mathsf{var}$~\cite{You:2016uy}, at its fixed
% point $\mathsf{fixp}$ we should have $\mathsf{fixp}=\mathpzc{T}_{\bm{\kappa}',j}
% (\mathsf{fixp},1)$. Combined with the uniqueness of the
% fixed point of an SIF, we could therefore draw the conclusion from the
% equation mentioned at the beginning that
% $\norm{\bm{\rho}(\bm{\alpha}(\bm{\kappa}'),\bm{\kappa}')}_{\infty}\bm{\alpha}(\bm{\kappa}')
% =\bar{\rho}\bm{\alpha}(\bm{\kappa}')$, which leads
% to $\norm{\bm{\rho}(\bm{\alpha}(\bm{\kappa}'),\bm{\kappa}')}_{\infty}=\bar{\rho}$
% directly.
\end{proof}

\begin{lemma}
    Either $\bm{\alpha}^{*}(\bm{\kappa}')=\bm{\alpha}(\bm{\kappa}')$ or 
    $\bm{\alpha}^{*}(\bm{\kappa}')=\hat{\bm{\alpha}}(\bm{\kappa}')$.
\label{lma:either}
\end{lemma}
\begin{proof}
Suppose $\mathsf{fixp}\in\mathbb{R}_{+}^{q}$ is the fixed point of
${{\mathpzc{T}_{\bm{\kappa}',j}
(\mathsf{var},1)}\big/\norm{\mathpzc{H}_{\bm{\kappa}'} (\mathsf{var},
1)}_{\bm{\kappa}'}}$, with respect to $\mathsf{var}$.  By Theorem~1
in~\cite{Krause:2001wd}), there exists $\lambda>0$, such that
${\mathpzc{T}_{\bm{\kappa}',j} (\mathsf{fixp},1)}=\lambda\mathsf{fixp}$, with
$\norm{\frac{1}{\bar{\rho}}\mathsf{fixp}}_{\bm{\kappa}'}=1$.  We then conclude
that, if in the limit in~\eqref{eq:alpha_kappa2} we have exactly $\nu=\lambda$,
then $\hat{\bm{\alpha}}(\bm{\kappa}')=\bm{\alpha}(\bm{\kappa}')$.  For clarity
in the following proof, we denote this specific
$\hat{\bm{\alpha}}(\bm{\kappa}')$ by $\hat{\bm{\alpha}}(\bm{\kappa}',\lambda)$.
Now consider any $\nu$ with $\nu>\lambda$. We look into the corresponding
sequence of the limit in~\eqref{eq:alpha_kappa2}. For any $k\geq 1$ and
$j\in\J$, denote $\alpha_j^{(k)}=\mathpzc{T}_{\bm{\kappa}',j}
(\bm{\alpha}^{(k-1)},\nu)$, with
$\bm{\alpha}^{(0)}=\bm{\alpha}(\bm{\kappa}',\lambda)$. Note that
$\alpha^{(0)}_{j}=\mathpzc{T}_{\bm{\kappa}',j} (\bm{\alpha}^{(0)},\lambda)$
for any $j\in\J$.  By our construction, the sequence
$\bm{\alpha}^{(0)},\bm{\alpha}^{(1)},\ldots,\bm{\alpha}^{(\infty)}$ converges
to $\bm{\alpha}(\bm{\kappa}')$, i.e.
$\bm{\alpha}^{(\infty)}=\bm{\alpha}(\bm{\kappa}')$. By our definition of the
sequence above, $\alpha_j^{(1)}=\mathpzc{T}_{\bm{\kappa}',j}
(\bm{\alpha}^{(0)},\nu)$. Meanwhile, note that for any $j\in\J$,
$\mathpzc{T}_{\bm{\kappa}',j}(\bm{\alpha}^{(0)},\nu)\geq
\mathpzc{T}_{\bm{\kappa}',j}
(\bm{\alpha}^{(0)},\lambda)$ holds, due to
that $\mathpzc{T}_{\bm{\kappa}',j} (\bm{\alpha}^{(0)},\nu)$ is monotonic in
$\nu$ (or $\lambda$). Therefore, $\bm{\alpha}^{(1)}\geq\bm{\alpha}^{(0)}$.  By using the same
way, one can prove that $\bm{\alpha}^{(\infty)}\geq\bm{\alpha}^{(0)}$, and hence
$\hat{\bm{\alpha}}(\bm{\kappa}')\geq\bm{\alpha}(\bm{\kappa}')$.  Similarly, for
any $\nu<\lambda$, we can prove that
$\hat{\bm{\alpha}}(\bm{\kappa}')\leq\bm{\alpha}(\bm{\kappa}')$ holds.  Hence the
conclusion.
\end{proof}

\begin{lemma}
Denote  
$\bm{\eta}^{*}(\bm{\kappa}')=\argmax_{\bm{\eta}}\{\norm{\bm{\eta}}_1:
\eqref{eq:coupling},\eqref{eq:limited},\bm{\alpha}=
\bm{\alpha}^{*}(\bm{\kappa}'),\bm{\kappa}=\bm{\kappa}'\}$. 
Then $\eta^{*}_{1}(\bm{\kappa}')=\eta_{2}^{*}(\bm{\kappa}')=\cdots=\eta_{q}^{*}(\bm{\kappa}')$.
\label{lma:eta}
\end{lemma}
\begin{proof}
We first note that the feasibility holds regarding the constraints in the
optimization problem stated in this lemma, by setting
$\eta_1=\eta_2=\cdots=\eta_q=0$.  The we prove that at its optimum we have
$\eta_1=\eta_2=\cdots=\eta_q$.  By Lemma~\ref{lma:either}, we either have
$\bm{\alpha}^{*}(\bm{\kappa}')=\bm{\alpha}(\bm{\kappa}')$ or
$\bm{\alpha}^{*}(\bm{\kappa}')=\hat{\bm{\alpha}}(\bm{\kappa}')$.  We first
consider the case of
$\bm{\alpha}^{*}(\bm{\kappa}')=\hat{\bm{\alpha}}(\bm{\kappa}')$. In this case,
we have for any $j\in\J$ that
$\eta_j\leq\hat{\alpha}_{j}(\bm{\kappa}')C_j(\hat{\bm{\alpha}}(\bm{\kappa}'),
\bm{\kappa}')\big/\int_{T}t_j(\tau)\diff\tau$, by the constraint~\eqref{eq:coupling}.
Combined with~\eqref{eq:alpha_kappa2} that
$\nu=\hat{\alpha}_{j}(\bm{\kappa}')C_j(\hat{\bm{\alpha}}(\bm{\kappa}'),\bm{\kappa}')\big/
\int_{T}t_j(\tau)\diff\tau$
($j\in\J$), we then have
$\eta_j\leq\nu=\min_{i\in\R}\{c_i\big/\int_{T}\sum_{j\in\J_i}t_j(\tau)\diff\tau\}$
for any $j\in\J$. Also, note that $\eta_j\leq\nu$ with all $j\in\J$ indicates
that the fronthaul capacity constraint in~\eqref{eq:limited} is satisfied. In
addition, by the definition of $\bm{\alpha}^{*}(\bm{\kappa}')$
in~\eqref{eq:opt_alpha}, we have
$\bm{\alpha}^{*}(\bm{\kappa}')\leq\bm{\alpha}(\bm{\kappa}')$, which leads to
$\bm{\rho}(\bm{\alpha}^{*}(\bm{\kappa}'),\bm{\kappa}')\leq
\bm{\rho}(\bm{\alpha}(\bm{\kappa}'),\bm{\kappa}')$. Along with
Lemma~\ref{lma:bar_rho}, we conclude that
$\rho_{i}(\bm{\alpha}^{*}(\bm{\kappa}'),\bm{\kappa}')\leq \bar{\rho}$ holds for
any $i\in\R$, and hence the maximum resource limit constraint
in~\eqref{eq:limited} is satisfied. In this case, we should have
$\eta_1=\eta_2=\cdots=\eta_q=\nu$ so as to reach the maximum of
$\norm{\bm{\eta}}_1$.

For the other case that
$\bm{\alpha}^{*}(\bm{\kappa}')=\bm{\alpha}(\bm{\kappa}')$, we can verify that the
maximum resource limit constraint in~\eqref{eq:limited} is satisfied, 
by applying Lemma~\ref{lma:bar_rho} directly. Also, from the proof of
Lemma~\ref{lma:eta}, we have some $\lambda$ such that 
$\lambda=\alpha_1(\bm{\kappa}')/
\mathpzc{T}_{\bm{\kappa}',1}(\bm{\alpha}(\bm{\kappa}'),1)=\alpha_2(\bm{\kappa}')/
\mathpzc{T}_{\bm{\kappa}',2}(\bm{\alpha}(\bm{\kappa}'),1)=\cdots=\alpha_q(\bm{\kappa}')/
\mathpzc{T}_{\bm{\kappa}',q}(\bm{\alpha}(\bm{\kappa}'),1)$. By the proof in
Lemma~\ref{lma:either}, we know that in this case we have $\lambda<\nu$ holds,
and thus the fronthaul capacity constraint in~\eqref{eq:limited} is satisfied. 
By the constraint in~\eqref{eq:coupling}, we have
$\eta_j\leq{\alpha}_{j}(\bm{\kappa}')C_j({\bm{\alpha}}(\bm{\kappa}'),
\bm{\kappa}')\big/\int_{T}t_j(\tau)\diff\tau=\lambda$. We therefore have
$\eta_1=\eta_2=\cdots=\eta_{q}=\lambda$ so as to reach the maximum of
$\norm{\bm{\eta}}_{1}$.
\end{proof}

% \begin{lemma}
% $\alpha^{*}_j(\bm{\kappa}')\geq\int_{T}{\eta^*_{j}t_j(\tau)}/
% {C_j(\bm{\alpha}^{*}(\bm{\kappa}'),\bm{\kappa}')}\diff\tau$, $j\in\J$.
% \end{lemma}

% \begin{lemma}
% $\sum_{j\in\J_i}\int_{T}\eta^{*}_{j}t_j(\tau)\leq c_i$, $j\in\J$.
% \end{lemma}

% \begin{lemma}
% $\rho_i(\bm{\alpha}^{*}(\bm{\kappa}'),\bm{\kappa}')\leq\bar{\rho}$,
% $i\in\R$.
% \end{lemma}

\begin{lemma}
For any $j\in\J$ in Lemma~\ref{lma:eta},~\eqref{eq:eta} holds. 
\begin{equation}
    \eta^{*}_{j}(\bm{\kappa}')=\min\left\{\alpha_j^{*}(\bm{\kappa}')\bigg/
    \mathpzc{T}_{\bm{\kappa}',j}(\bm{\alpha}^{*}(\bm{\kappa}'),1),~
\nu \right\}
\label{eq:eta}
\end{equation}
\end{lemma}
\begin{proof}
The proof is directly based on Lemma~\ref{lma:either} and~\ref{lma:eta} and
their corresponding proofs. For the two cases
$\bm{\alpha}^{*}(\bm{\kappa}')=\bm{\alpha}(\bm{\kappa}')$ and
$\bm{\alpha}^{*}(\bm{\kappa}')=\hat{\bm{\alpha}}(\bm{\kappa}')$, note that we
have $\eta_j=\min\{\lambda,\nu\}$ respectively, for all $j\in\J$. Since
$\lambda=\alpha_1(\bm{\kappa}')/
\mathpzc{T}_{\bm{\kappa}',1}(\bm{\alpha}(\bm{\kappa}'),1)=\alpha_2(\bm{\kappa}')/
\mathpzc{T}_{\bm{\kappa}',2}(\bm{\alpha}(\bm{\kappa}'),1)=\cdots=\alpha_q(\bm{\kappa}')/
\mathpzc{T}_{\bm{\kappa}',q}(\bm{\alpha}(\bm{\kappa}'),1)$, we reach the
conclusion.
\end{proof}

\subsection*{\textbf{The proof of Theorem~\ref{thm:opt_alpha}} is as follows:}
\begin{proof}
    % {\small (Proof of Theorem~\ref{thm:opt_alpha}).}
For the problem $P1:\argmax_{\bm{\alpha}}\{\norm{\bm{\eta}}_1:
\eqref{eq:coupling}-\eqref{eq:fairness}, \bm{\kappa}=\bm{\kappa}'\}$, we can
prove by~Lemma~\ref{lma:eta} that $\bm{\alpha}=\bm{\alpha}^{*}(\bm{\kappa}')$
satisfies the constraint~\eqref{eq:fairness}, due to that
$\eta^{*}_{1}(\bm{\kappa}')=\eta_{2}^{*}(\bm{\kappa}')=\cdots=\eta_{q}^{*}(\bm{\kappa}')$.
Also, $\bm{\alpha}^{*}(\bm{\kappa}')$ is the optimal solution for the problem 
$P2:\argmax_{\bm{\eta}}\{\norm{\bm{\eta}}_1:
\eqref{eq:coupling},\eqref{eq:limited},\bm{\alpha}=
\bm{\alpha}^{*}(\bm{\kappa}'),\bm{\kappa}=\bm{\kappa}'\}$, according
to~Lemma~\ref{lma:eta}, which means that $\bm{\alpha}(\bm{\kappa}')$ satisfies
the constraints~\eqref{eq:coupling} and~\eqref{eq:limited}. Therefore,
$\bm{\alpha}(\bm{\kappa}')$ is a feasible solution of $P1$. Consider there are
another feasible solution of $P1$, $\langle \bm{\eta}', \bm{\alpha}' \rangle$
that leads to an objective
value $\norm{\bm{\eta}'}>\norm{\bm{\eta}^{*}}$. To meet the constraint
in~\eqref{eq:fairness}, we should have $\eta'_{1}=\eta'_{2}=\cdots\eta'_{q}$.
Therefore, we have either $\eta'_j=\beta\lambda$ or $\eta'_j=\beta\nu$ ($\beta>1$,
$j\in\J)$. 
In addition, under the constraints~\eqref{eq:coupling} and~\eqref{eq:limited}, we should
respectively have $\eta'_j\leq{\alpha}'_{j}C_j(\bm{\alpha}',
\bm{\kappa}')\big/\int_{T}t_j(\tau)\diff\tau$ and
$\eta_j\leq\min_{i\in\R}\{c_i\big/\int_{T}\sum_{j\in\J_i}t_j(\tau)\diff\tau\}$.
It can then be verified that the two constraints cannot be satisfied
together by $\bm{\eta}'$, which conflicts our assumption that $\langle
\bm{\eta}',\bm{\alpha}'\rangle$ is feasible. Hence the conclusion.
\end{proof}

\subsection{CoMP-cell Selection}

The partial optimality condition of CoMP-cell selection is given by
Theorem~\ref{thm:opt_condition}, which is proved based on Lemma~\ref{lma:H1} and
Lemma~\ref{lma:H2}.

\begin{definition}
For any $j\in\J$ and $\mathpzc{r}\subseteq\R$, define the mapping $\mathpzc{E}_{j,g}:{\{0,1\}}^{m\times q}\rightarrow{\{0,1\}}^{m\times q}:
\bm{\kappa}\rightarrow\hat{\bm{\kappa}}|\I_{j}(\hat{\bm{\kappa}})
=\I_{j}(\bm{\kappa})\cup\mathpzc{r}$.
\end{definition}

\begin{definition}[\small{CoMP-cell filter}] 
    Denote $\bm{\kappa}''=\mathpzc{E}_{j,\{i\}}(\bm{\kappa}')$.
    For any UE $j$ and any target set
    $\mathpzc{r}\subseteq\R$, define the mapping
    ${\{0,1\}}^{m\times q}\rightarrow{\{0,1\}}^{m\times q}$ as the filter of
    $\mathpzc{r}$
\begin{equation}
    \mathpzc{F}_{j,\{i\}}:\bm{\kappa}'\mapsto
    \left\{\begin{array}{ll}
        \bm{\kappa}'' & \mathsf{condition}
        \textnormal{ satisfied} \\
        \bm{\kappa}' & \textnormal{otherwise}
    \end{array}\right.
\label{eq:funcF}
\end{equation}
such that the $\mathsf{condition}$ is
$\rho_i(\mathpzc{H}_{\bm{\kappa}''}(\bm{\alpha}^{*}(\bm{\kappa}'),\eta^{*}_j),\bm{\kappa}'')\leq\rho_i(\bm{\alpha}^{*}(\bm{\kappa}'),\bm{\kappa}'')$,
where the parameter $\mu_i$ is defined as
$\mu=\mathpzc{T}_{\bm{\kappa}'',j}
    (\bm{\alpha}^{*}(\bm{\kappa}'),\eta^{*}_j)
    \big/\alpha^{*}_{j}(\bm{\kappa'})$
and the resource allocation $\bm{\alpha}(\bm{\kappa}')$ follows the definition in~\eqref{eq:opt_alpha}.
\end{definition}

\begin{lemma}
    $\mathpzc{H}_{\bm{\kappa}''}(\bm{\alpha}^{*}(\bm{\kappa}'),\eta^{*}_j(\bm{\kappa}'))\leq\bm{\alpha}^{*}(\bm{\kappa}')$.
\label{lma:H1}
\end{lemma}
\begin{proof}
Let $\bm{\kappa}''=\mathpzc{E}_{j,\{i\}}(\bm{\kappa}')$. Then we have
$\I_j(\bm{\kappa}')\subset\I_j(\bm{\kappa}'')$, which leads to that
$\gamma_j(\bm{\kappa}'')>\gamma_j(\bm{\kappa}')$ according to~\eqref{eq:sinr}.
Therefore, we have $\mathpzc{T}_{\bm{\kappa}'',j}
(\bm{\alpha}^{*}(\bm{\kappa}'),\eta^{*}_j)
<\alpha^{*}_{j}(\bm{\kappa'})$. Hence the conclusion.
Since $\I_{k\neq j}(\bm{\kappa}'')=\I_{k\neq j}(\bm{\kappa}')$, we conclude
that $\mathpzc{T}_{\bm{\kappa}'',k\neq j}
(\bm{\alpha}^{*}(\bm{\kappa}'),\eta^{*}_{k\neq j})
=\alpha^{*}_{k\neq j}(\bm{\kappa'})$. Hence the conclusion.
\end{proof}

\begin{lemma}
    $\lim_{k\rightarrow\infty}\mathpzc{H}^{(k)}_{\bm{\kappa}''}(\bm{\alpha}_{\geq
    0},\eta^{*}_j(\bm{\kappa}'))\leq\bm{\alpha}^{*}(\bm{\kappa}')$ with $\bm{\alpha}_{\geq
    0}\in\mathbb{R}_{+}^{q}$.
\label{lma:H2}
\end{lemma}
\begin{proof}
    Since $\mathpzc{H}^{(k)}_{\bm{\kappa}''}(\mathsf{var},\eta^{*}_j)$ is
    monotonic in $\mathsf{var}$, we reach the conclusion by Lemma~\ref{lma:H1}.
\end{proof}

\begin{theorem}
[\small{Partial Optimality of $\mathpzc{F}_{j,\{i\}}$}]
For any $\bm{\kappa}'$, $j\in\J$, and $\mathpzc{r}\subseteq\R$, $\max_{\bm{\alpha}}\{\norm{\bm{\eta}}_1:
\eqref{eq:coupling}-\eqref{eq:fairness},
\bm{\kappa}=\mathpzc{F}_{j,\{i\}}(\bm{\kappa}')\}\geq\max_{\bm{\alpha}}\{\norm{\bm{\eta}}_1:\eqref{eq:coupling}-\eqref{eq:fairness}, \bm{\kappa}=\bm{\kappa}'\}$.
\label{thm:opt_condition}
\end{theorem}
\begin{proof}
If $\mathpzc{F}_{j,\{i\}}(\bm{\kappa}')=\bm{\kappa}'$, then we have that the
equality in the statement holds
element-wisely for $\bm{\alpha}$ in the theorem. For the other case that
$\mathpzc{F}_{j,\{i\}}(\bm{\kappa}')=\bm{\kappa}''$, we have by
Lemma~\ref{lma:H1} that
$\eta_j^{*}(\bm{\kappa}')\leq\bm{\alpha}^{*}(\bm{\kappa}')/\mathpzc{T}_{\bm{\kappa}'',
j}(\bm{\alpha}^{*}(\bm{\kappa}'),1)$.
% Now we consider the fixed-point iterations of
% $\alpha_{j}^{(k)}=\mathpzc{T}_{\bm{\kappa}'',j}(\bm{\alpha}^{(k-1)}, 
% \eta_j^{*}(\bm{\kappa}'))$, with
% $\bm{\alpha}^{(0)}=\bm{\alpha}^{*}(\bm{\kappa}')$. By Lemma~\ref{lma:H1} and the
% monotonicity, we
% have
% $\alpha_j^{(k-1)}/\mathpzc{T}_{\bm{\kappa}'',j}(\bm{\alpha}^{(k-1)},1)\leq\alpha_j^{(k-1)}/\mathpzc{T}_{\bm{\kappa}'',j}(\bm{\alpha}^{(k)},1)$
% for all $k\geq 1$.
Thus all the active constraints in~\eqref{eq:coupling} and~\eqref{eq:limited}
are relaxed. Combined with Lemma~\ref{lma:H2}, we have
$\eta_j(\bm{\kappa}'')\geq\eta_j(\bm{\kappa}')$. Hence the conclusion.
\end{proof}

\subsection{Algorithm Design}

The proposed algorithm is shown in Algorithm~\ref{alg:alg}, of which the
time-consuming part is on computing the convergence point
$\bm{\alpha}^{*}(\bm{\kappa})$ of equation~\eqref{eq:alpha_kappa}
and~\eqref{eq:alpha_kappa2}, with respect to the CoMP association pattern
$\bm{\kappa}$. As stated in Line~\ref{alg:l:kappa}, the fixed-point iterations
are done with $\sum_{j=1}^{q}|\R_{(\ell(j))}|$ rounds. Suppose the
time-complexity of the algorithm for computing the fixed point
of~\eqref{eq:alpha_kappa} and~\eqref{eq:alpha_kappa2} is in $O(K)$. Then the
time-complexity of Algorithm~\ref{alg:alg} is in $O(qmK)$.  
\begin{algorithm}[h!b]
\KwIn{
$p_i$, $\bm{\kappa}'$, $\bm{h}_{ij}$, $\bm{w}_i$, and $t_j(\tau)$, for $i\in\R, j\in\J$
}
\KwOut{
    $\bm{\kappa}''$,
    $\bm{\alpha}(\bm{\kappa}'')$, and $\bm{\eta}(\bm{\kappa}'')$.
}
Let $r(i,j)$ be the $j_{\textnormal{th}}$ element in the set $\R_{\ell(i)}$. \;
$\bm{\kappa}''=\mathpzc{F}_{2,\{r(1,1)\}}\circ\mathpzc{F}_{1,\{r(1,2)\}}\circ\cdots\circ\mathpzc{F}_{q,\{r(q,|\R_{\ell(q)}|)\}}(\bm{\kappa'})$
\label{alg:l:kappa} \;
$\bm{\alpha}(\bm{\kappa}'')$ and $\bm{\eta}(\bm{\kappa}'')$ are computed
by~\eqref{eq:opt_alpha} and~\eqref{eq:eta}, respectively.
\caption{CoMP-cell selection and resource allocation}
\label{alg:alg}
\end{algorithm}

\section{Simulation}
\label{sec:simulation}

We deploy $3$ C-RAN clusters with hexagonal coverage region ($500$ meters
radius). There is an RCC located in the center of each hexagon, along with a BBU
pool. In each hexagonal region, several BSs are deployed with an RRH\@. The RRH
is connected with a Common Public Radio Interface (CPRI) based fronthaul of $2.5$ Gbps
capacity limit to each RRH~\cite{Anonymous:vkX0Zmz-}. Multiple UEs are randomly and
uniformly distributed in each hexagonal region. The network operates at $2$ GHz.
Each RU is set to $180$ KHz bandwidth and the bandwidth for each cell is $20$
MHz. We remark that the simulation setting of bandwidth follows the 3GPP
standardization document~\cite{release13-1}. The noise power spectral density is
set to $-174$ dBm/Hz. The transmit power of BSs on each RU is $200$ mW. The path
loss between BS and UE follows the standard 3GPP micro models~\cite{release9}.
The shadowing coefficients are generated by the log-normal distribution with $3$
dB standard deviation~\cite{release9}. 

\begin{figure}[!ht]
\vspace{-0.0cm}
    \centering
    \includegraphics[width=\linewidth]{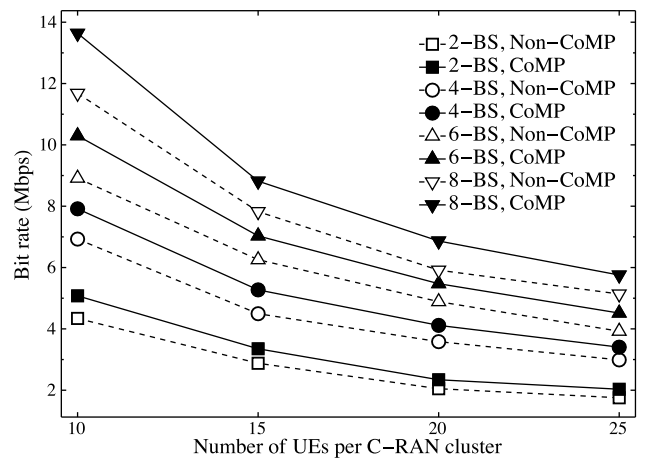}
    \caption{Bit rate vs.~number of UEs ($\bar{\rho}=1.0$).}
\vspace{-0.0cm}
\label{fig:Rate_UE}
\end{figure}

In the numerical results below, we compare the QoS performance of CoMP with that
of non-CoMP case.  In both cases, the time-frequency resource allocation is
optimized by Theorem~\ref{thm:opt_alpha}. The evaluation is done under the
scenarios with different number of BSs, UEs, and values of the maximum resource
limit $\bar{\rho}$.  \figurename~\ref{fig:Rate_UE} shows the QoS performance
with respect to different amount of UEs and BSs in each C-RAN cluster. We deploy
respectively $2$, $4$, $6$, $8$ BSs (RRHs) that are capable to cooperatively serve UEs via
CoMP within each C-RAN cluster. Intuitively, the QoS depends on user density and
resource sharing.  Numerically, the QoS decreases with the increase
of UEs' density.  Besides, the QoS can be enhanced by deploying more
cooperative RRHs in a C-RAN cluster. Compared to the non-CoMP case, the QoS
performance always benefits from optimizing the CoMP-cell selection. When UEs are
densely distributed, one can achieve almost the same QoS enhancement by CoMP within C-RAN
cluster, as by increasing the BS density. With the increase of BS density, the UEs
gain more on QoS improvement. On average, the QoS is improved via CoMP by
$11.6\%$. In \figurename~\ref{fig:Rate_BS}, we compare the QoS performance for
the cases with different maximum available resource constraint, i.e.
$\bar{\rho}$. The network benefits more via CoMP with a larger value of
$\bar{\rho}$. That means, the cooperation among BSs via CoMP
would be crucial, with sufficient available time-frequency resource in the
network. On average, the QoS improvement of CoMP can reach to $11.3\%$.

In general, the possible improvement on QoS through CoMP and resource allocation
is sensitive to the network density (for both BS and UE) as well as the resource
limit in each cell. In 5G, the network is likely to be ultra-densely deployed
with BSs (e.g.~small/femto stations), resulting in that there could be more
stations than UEs in a cell coverage of 5G cellular networks.  Further, as the
mobile systems of 5G are broadening their spectrum, the available resource would
be more sufficient in the next generation networks. 
% , compared to the current situation. Based on the
% numerical results, we conclude that joint optimization on CoMP-cell selection and resource
% allocation plays an important role with C-RAN architectures.

\begin{figure}[!ht]
\vspace{-0.0cm}
    \centering
    \includegraphics[width=\linewidth]{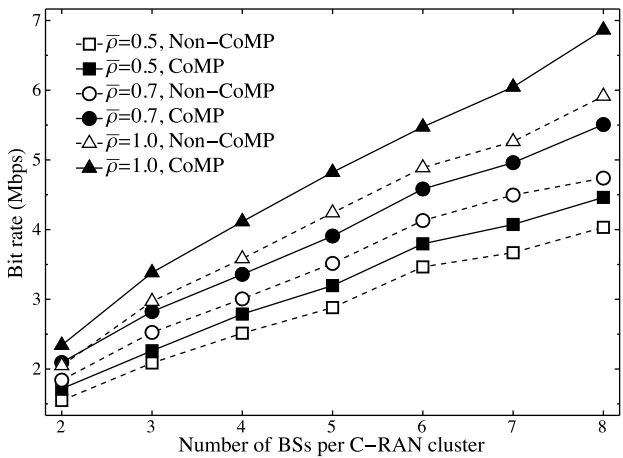}
    \caption{Bit rate vs.~number of BSs ($20$ UEs per cluster).}
\vspace{-0.0cm}
\label{fig:Rate_BS}
\end{figure}

\section{Conclusion}
\label{sec:conclusion}

For the joint CoMP-cell selection and resource allocation problem with
fronthaul-constrained C-RAN, theoretical analysis regarding the computational
complexity has been given. A joint optimization problem of cell selection and
resource allocation has been proposed. 

\section*{Acknowledgement}
This work has been supported by the Swedish Research Council and the
Link\"{o}ping-Lund Excellence Center in Information Technology (ELLIIT), Sweden,
and the European Union Marie Curie project MESH-WISE (FP7-PEOPLE-2012-IAPP\@:
324515), DECADE (H2020-MSCA-2014-RISE\@: 645705), and WINDOW
(FP7-MSCA-2012-RISE\@: 318992). The work of D. Yuan has been
carried out within European FP7 Marie Curie IOF project 329313.

\bibliographystyle{IEEEtran}
\bibliography{ref}
\end{document}